\documentclass[11pt,reqno,section]{amsart}

\usepackage[english]{babel}
\usepackage[utf8]{inputenc}
\usepackage{amsmath}
\usepackage{amsthm}
\usepackage{verbatim}
\usepackage{enumerate}
\usepackage{amssymb}
\usepackage{amsfonts}
\usepackage{amscd,bezier}
\usepackage{anysize}
\usepackage{indentfirst}
\usepackage{upref}
\usepackage{color}
\usepackage[pagewise]{lineno}
\usepackage{tikz}
\usepackage{theoremref}
\usetikzlibrary{matrix}
\usepackage{pstricks}
\usepackage{subfig}
\usepackage{graphicx}
\usepackage{pstricks}
\usepackage{amscd}
\usepackage{relsize}
\usepackage[toc,page]{appendix}
\usepackage{commath}

\allowdisplaybreaks

\numberwithin{equation}{section}

\newtheorem{theorem}{Theorem}[section]
\newtheorem{lemma}[theorem]{Lemma}

\newtheorem{definition}[theorem]{Definition}
\newtheorem{proposition}[theorem]{Proposition}
\newtheorem{corollary}[theorem]{Corollary}

\newtheorem{example}[theorem]{Example}

\newtheorem{remark}[theorem]{Remark}

\marginsize{3cm}{3cm}{3cm}{3cm}

\newcommand{\NmuD}{{\mathrm{N}\mu\mathrm{D}}}
\newcommand{\sNmuD}{{\mathrm{sN}\mu\mathrm{D}}}

\newcommand{\muD}{{\mu\mathrm{D}}}
\newcommand{\Nmu}{{\mathrm{N}\mu}}
\newcommand{\NED}{\mathrm{NED}}

\newcommand{\NPD}{\mathrm{NPD}}
\newcommand{\gor}{\mathfrak{Gor}}
\newcommand{\st}{\mathfrak{st}}
\newcommand{\un}{\mathfrak{un}}
\newcommand{\St}{\mathrm{St}}
\newcommand{\Un}{\mathrm{Un}}
\newcommand{\PD}{\mathrm{PD}}

\newcommand{\Si}{\Sigma}
\newcommand{\ED}{\mathrm{ED}}
\newcommand{\Id}{\mathrm{Id}}

\renewcommand{\P}{\mathrm{P}}

\newcommand{\Q}{\mathrm{Q}}

\newcommand{\sgn}{\mathrm{sgn}}

\numberwithin{equation}{section}

\allowdisplaybreaks
\usepackage{color}

\title[Nonuniform spectra $\&$ nonuniform kinematic similarity]{Spectrum invariance dilemma for nonuniformly kinematically similar systems}

\author[C. A. Gallegos]{Claudio A. Gallegos$^{*}$}
\address{Universidad de Chile, UCH, Facultad de Ciencias, Departamento de Matem\'aticas, Casilla 653, Santiago, Chile.}
\email{claudio.gallegos.castro@gmail.com, nestor.jara@ug.uchile.cl} 
\thanks{$^{*}$
	This author was partially supported by ANID/FONDECYT Postdoctorado No 3220147. }

\author[N. Jara]{N\'estor Jara$^{**}$}

\thanks{$^{**}$
	This author was partially supported by ANID, Beca de Doctorado Nacional 21220105.}

\date{}

\begin{document}
	
	\begin{abstract}
		We unveil instances where nonautonomous linear systems manifest distinct nonuniform $\mu$-dichotomy spectra despite admitting nonuniform $(\mu, \varepsilon)$-kinematic similarity. Exploring the theoretical foundations of this lack of invariance, we discern the pivotal influence of the parameters involved in the property of nonuniform $\mu$-dichotomy such as in the notion of nonuniform $(\mu, \varepsilon)$-kinematic similarity. To effectively comprehend these dynamics, we introduce the stable and unstable optimal ratio maps, along with the $\varepsilon$-neighborhood of the nonuniform $\mu$-dichotomy spectrum. These concepts provide a framework for understanding scenarios governed by the noninvariance of the nonuniform $\mu$-dichotomy spectrum.
	\end{abstract}
	

	%\dedicatory{}
	

	\maketitle

	\section{Introduction}

	Consider the  nonautonomous linear differential equation
	\begin{equation}\label{613}
	\dot{y}=A(t)y(t),
	\end{equation}
	where $t\mapsto A(t)\in\mathbb{R}^{d\times d}$ is a locally integrable matrix valued function. As usual, we denote by $\Phi\colon\mathbb{R}\times \mathbb{R}\to \mathbb{R}^{d\times d}$, $(t,s)\mapsto \Phi(t,s)$, the evolution operator --or the transition matrix-- associated to \eqref{613}. That means, for every $s\in\mathbb{R}$ and $x_0\in\mathbb{R}^{d}$, the function $\Phi(\cdot,s)x_0$ solves the initial value problem \eqref{613} with initial condition $y(s)=x_0$.

	\subsection{kinematic similarity and dichotomy spectrum}
	Inspired by the notion of dichotomy spectrum for skew product flows introduced in \cite{Sacker}, S. Siegmund defined the concept of dichotomy spectrum for system \eqref{613} as the set 
	\[
	\Si(A):=\{\gamma\in\mathbb{R}: \dot{y}=[A(t)-\gamma \Id]y(t) \text{ admits no exponential dichotomy}\},
	\]
	demonstrating that is a finite union of closed intervals (spectral intervals), which could be unbounded but with compactness ensured if system \eqref{613} exhibits the property of bounded growth, see \cite[Theorem~3.1]{Siegmund}. In the subsequent work of S. Siegmund \cite{Siegmund2}, leveraging the spectral dichotomy theory developed in \cite{Siegmund}, a pivotal outcome was achieved: a compelling reducibility result for system \eqref{613}, see \cite[Theorem~3.2]{Siegmund2}. This result elucidates that system \eqref{613} is kinematically similar to a system which is block-diagonalized, where each block --of dimension less than $d$-- precisely corresponds to a spectral interval of $\Si(A)$. This interesting result relies on a fundamental property: if two systems are kinematically similar, then they have the same dichotomy spectrum, see \cite[Corollary~2.1]{Siegmund2}.
	
	The formulation of the dichotomy spectrum incorporates the notion of exponential dichotomy for nonautonomous systems, which in simple terms emulates the idea of hyperbolic equilibrium from the autonomous framework, see e.g. \cite{Coppel,Perron}. In 2005, L. Barreira and C. Valls \cite{Barreira7}, motivated to consider broader categories of hyperbolic dynamics, introduced the notion of {\it nonuniform exponential dichotomy} for system \eqref{613}. This concept extends the classical exponential dichotomy, and henceforth, the academic community have adopted to add the term ``uniform" when referring to the classic one. We suggest to the reader the book \cite{BV}, where several contributions to the development of the nonuniform theory are presented.

	Following the ideas of S. Siegmund, the authors X. Zhang \cite{Xiang} and J. Chu {\it et al.} \cite{Chu} independently developed a spectral theory by considering the notion of nonuniform exponential dichotomy (slightly modificated definitions are utilized on each article). As in the classical case but considering the assumption of nonuniform bounded growth, it is proved that the nonuniform dichotomy spectrum can be decomposed as finite union of compact intervals, see \cite[Corollary~2.11]{Chu} and \cite[Theorem~2.1]{Xiang}. In addition, a reducibility theorem is obtained by considering a fundamental property: the invariance of nonuniform dichotomy spectra via nonuniform kinematic similarity.

	In a recent work, C. Silva \cite{Silva} introduced the notion of nonuniform $\mu$-dichotomy, providing a unifying framework that extends and integrates the discussed exponential dichotomy concepts (see Definition~\ref{NmuD} in Section 2). Considering this nonuniform $\mu$-dichotomy, a spectral theory is developed and a reducibility theorem is derived by once again leveraging an essential property: the invariance of the nonuniform $\mu$-dichotomy spectra through nonuniform $(\mu,\varepsilon)$-kinematic similarity.

	\subsection{Main results and contributions}
	The aim of this article is to defy the conventional idea of invariance of the nonuniform $\mu$-dichotomy spectra for nonautonomous linear systems which are nonuniformly $(\mu,\varepsilon)$-kinematically similar. This assertion is explained on Remark \ref{660} and supported by Examples~\ref{ex1} and \ref{634}, in which we exhibit nonuniformly $(\mu,\varepsilon)$-kinematically similar systems presenting distinct nonuniform $\mu$-dichotomy spectra. We emphasize that the lack of invariance of the nonuniform dichotomy spectrum does not depend on the specific notion of nonuniform dichotomy under consideration. In all cases, noninvariance is detected; see Remark~\ref{641}.
	
	In this article, we will provide a theoretical foundation for this noninvariance phenomenon. In this context, we observe that the parameters involved in the notions of nonuniform $\mu$-dichotomy, $(\Nmu,\epsilon)$-growth, and $(\mu,\varepsilon)$-kinematic similarity, significantly  contribute in these dynamics. For that reason, we introduce some novel concepts with the aim of appropriately managing the above-mentioned parameters:
	\begin{itemize}
		\item {\it The region of stable and unstable constants}, see Definition~\ref{regions}.
		\item On every spectral gap, we define {\it the maps of optimal stable and unstable ratio} along with its underlying properties, see Definition~\ref{optimalmaps} together with the subsequent results from Section~\ref{Sect2}. 
		\item We introduce the {\it $\varepsilon$-neighborhood of the nonuniform $\mu$-dichotomy spectrum} and the {\it $\varepsilon$-interior of the nonuniform $\mu$-resolvent} for system \eqref{613}, see Definition~\ref{epsilonspectra}. In addition, we elucidate the relationship between this concept of spectrum and the classic one for systems that admit nonuniform $(\mu,\varepsilon)$-kinematic similarity, see Corollary~\ref{608}. 
	\end{itemize}
	
	By introducing this new concepts, our objective is to enhance the understanding in scenarios characterized by the noninvariance of the nonuniform $\mu$-dichotomy spectrum.

	\subsection{Consequences about the noninvariance of the nonuniform spectra}

	Reducibility results are relevant since they allow for an easier interpretation of nonautonomous equations, but also as a crucial initial step toward several results. For instance, a remarkable application of these type of construction, which typically needs dichotomy spectrum preservation, is a normal form result, which states that under a polynomial transformation, it is possible to eliminate suitable Taylor terms of a nonlinear perturbation. The first of these results for the nonautonomous framework was given by S. Siegmund in \cite{Siegmund3}.
	
	Furthermore, even a normal form result is just a tool in order to find a local smooth linearization, as given originally by H. Poincaré \cite{Poincare} for autonomous analytic complex equations, and by S. Sternberg \cite{Sternberg1,Sternberg2} for real equations. In the nonautonomous context, L. V. Cuong, S. T. Doan and S. Siegmund \cite{Cuong} were able to find a local smooth linearization for systems admitting the classic exponential dichotomy. Once again, for this construction, the invariance of the dichotomy spectrum is required.

	\section{Optimal dichotomy constants}\label{Sect2}
	
	In this section, we begin by recalling key concepts related to the nonuniform $\mu$-dichotomy presented in \cite{Silva}. Following this, we introduce {\it the optimal stable and unstable ratio functions}, providing a detailed exposition of their fundamental properties.

	\begin{definition}%\label{361}    
		A strictly increasing function $\mu\colon\mathbb{R}\to (0,+\infty)$ is said to be a \textbf{growth rate}, if $\mu(0)=1$, $\mu(t)\to +\infty$ as $t\to +\infty$, and $\mu(t)\to 0$ as $t\to -\infty$. Moreover, if  $\mu$ is differentiable, it is said a \textbf{differentiable growth rate}.
	\end{definition}

	\begin{definition}\label{NmuD}
		Let $\mu:\mathbb{R}\to (0,+\infty)$ be a growth rate. The system \eqref{613} admits \textbf{nonuniform $\mu$-dichotomy} if  there exist an invariant projector $t\mapsto \P(t)\in\mathbb{R}^{d\times d}$, that means $\P$ satisfies 
		\[
		\P(t)\Phi(t,s)=\Phi(t,s)\P(s),\qquad \text{for all $t,s\in\mathbb{R}$,}
		\]
		constants $K\geq 1$, $\alpha<0$, $\beta>0$ and $\theta,\nu\geq 0$ satisfying $\alpha+\theta<0$ and $\beta-\nu>0$, such that  
		\begin{equation}\label{600} 
		\left\{ \begin{array}{lc}
		\|\Phi(t,s)\P(s)\|\leq K
		\left(\dfrac{\mu(t)}{\mu(s)}\right)^\alpha\mu(s)^{\sgn(s)\theta}, &\forall\,\, t\geq s ,\\
		\\ \|\Phi(t,s)[\Id-\P(s)]\|\leq K
		\left(\dfrac{\mu(t)}{\mu(s)}\right)^\beta\mu(s)^{\sgn(s)\nu}, &\forall\,\,  t\leq s.
		\end{array}
		\right.
		\end{equation}
		Moreover, if $\theta=\nu=0$, then we say the system \eqref{613} admits \textbf{uniform $\mu$-dichotomy}.\end{definition}
	
	Throughout the text, we abbreviate the notions of nonuniform $\mu$-dichotomy by $\NmuD$ and uniform $\mu$-dichotomy by $\muD$. Usually, when we refer to a system having $\NmuD$, we will emphasize that admits $\NmuD$ with parameters $(\P;\alpha,\beta,\theta,\nu)$. We omit $K$ in the parameters because is not fundamental for our later purposes. For the case $\P=\Id$, there are no constants $\beta,\nu$, hence we denote $(\Id;\alpha,*,\theta,*)$ and for $\P=0$ we write $(0;*,\beta,*,\nu)$.
	
	The preceding concept broadens various dichotomies based on the growth rate $\mu$ and the parameters involved in the estimations:
	\begin{itemize}
		\item Define $\mu(t)=e^t$, for all $t\in\mathbb{R}$. Consider $\beta=-\alpha$ and $\theta=\nu=0$. In this case we recover the notion of {\it exponential dichotomy} ($\ED$) defined by O. Perron in \cite{Perron}. Moreover, in the absence of constraints on the parameters, we arrive at the concept of {\it nonuniform exponential dichotomy} ($\NED$), a topic extensively studied in \cite{Barreira2,Chu,Dragicevic2,Xiang}.
		
		\item Consider an arbitrary strictly increasing surjective function $\nu\colon[0,+\infty)\to[1,+\infty)$. This function $\nu$ induces a growth rate $\mu$ defined as
		\begin{equation*}
		\mu(t):= \left\{ \begin{array}{lcc}
		\nu(t) &  \text{ if } &   t\geq 0, \\
		\\ \frac{1}{\nu(|t|)} & \text{ if }& t\leq 0 .
		\end{array}
		\right.
		\end{equation*}
		For instance, by considering a function $\nu$ defined by $\nu(t)= t+1$, for all $t\geq0$, we obtain a growth rate $\mu$ associated to the {\it nonuniform polynomial dichotomy} ($\NPD$). If moreover $\theta=\nu=0$, we obtain the {\it polynomial dichotomy} ($\PD$), see e.g.  \cite{Dragicevic3,Dragicevic5,Dragicevic6}.
	\end{itemize}

	For $\gamma\in \mathbb{R}$ and a differentiable growth rate $\mu$, consider the system
	\begin{equation}\label{615}
	\dot{y}=\left[A(t)-\gamma\frac{\mu'(t)}{\mu(t)}\Id\right]y(t),
	\end{equation}
	where $t\mapsto A(t)$ is the matrix valued function from system \eqref{613}. We said that system \eqref{615} is the {\it $\gamma$-shifted} system associated to \eqref{613}. Clearly, the transition matrix $\Phi_\gamma$ of system \eqref{615} is given by 
	\[
	\Phi_\gamma(t,s)=\Phi(t,s)\left(\frac{\mu(t)}{\mu(s)}\right)^{-\gamma}, \quad t,s\in\mathbb{R},
	\]
	where $\Phi(t,s)$ is the transition matrix of system \eqref{613}. We call $\Phi_{\gamma}(t,s)$ as the {\it $\gamma$-shifted operator}.
	
	The $\gamma$-shifted system \eqref{615} serves a crucial role to the following  concept, see \cite[p. 623]{Silva}.
	
	\begin{definition}%\label{423}
		Let $\mu\colon\mathbb{R}\to (0,+\infty)$ be a differentiable growth rate. The \textbf{nonuniform $\mu-$dichotomy spectrum} of system \eqref{613} is the set defined by
		\[
		\Si_\NmuD(A):=\left\{\gamma\in\mathbb{R}: \text{ \eqref{615} does not admit }\NmuD\right\}\,.
		\]
		
		Moreover, the \textbf{nonuniform $\mu$-resolvent} of system \eqref{613} is the set defined by 
		\[
		\rho_\NmuD(A):=\mathbb{R}\setminus \Si_\NmuD(A),
		\]
		i.e. the complement set of $\Si_\NmuD(A)$. 
	\end{definition}
	
	\begin{remark}
		{\rm Analogously, the \textbf{uniform $\mu-$dichotomy spectrum} and the \textbf{uniform $\mu$-resolvent} of system \eqref{613} is defined as above by simply replacing the $\NmuD$ by $\muD$.}
	\end{remark}

	For $\gamma\in \rho_{\NmuD}(A)$, we know the $\gamma$-shifted system admits nonuniform  $\mu$-dichotomy. In general, we do not know the constants involved in said dichotomy and they are by no means unique. Indeed, if the $\gamma$-shifted system \eqref{615} admits $\NmuD$ with parameters $(\P;\alpha,\beta,\theta,\nu)$ satisfying
	\begin{equation*}
	\left\{ \begin{array}{lc}
	\|\Phi_\gamma(t,s)\P(s)\|\leq K
	\left( \mathlarger{\frac{\mu(t)}{\mu(s)}}\right)^{\alpha}\mu(s)^{\sgn(s)\theta}, &\forall\,\, t \geq s \\
	\\ \|\Phi_\gamma(t,s)[\Id-\P(s)]\|\leq K
	\left( \mathlarger{\frac{\mu(t)}{\mu(s)}}\right)^{\beta}\mu(s)^{\sgn(s)\nu}, &\forall\,\,  t\leq s,
	\end{array}
	\right.
	\end{equation*}
	then for every $\widetilde{\alpha}\in (\alpha,-\theta)$, $\widetilde{\beta}\in (\nu,\beta)$, $\widetilde{\theta}\in (\theta,-\widetilde{\alpha})$, and $\widetilde{\nu}\in (\nu,\widetilde{\beta})$, the following estimations are verified as well
	\begin{equation*}
	\left\{ \begin{array}{lc}
	\|\Phi_\gamma(t,s)\P(s)\|\leq K
	\left( \mathlarger{\frac{\mu(t)}{\mu(s)}}\right)^{\widetilde{\alpha}}\mu(s)^{\sgn(s)\widetilde{\theta}}, &\forall\,\, t \geq s \\
	\\ \|\Phi_\gamma(t,s)[\Id-\P(s)]\|\leq K
	\left( \mathlarger{\frac{\mu(t)}{\mu(s)}}\right)^{\widetilde{\beta}}\mu(s)^{\sgn(s)\widetilde{\nu}}, &\forall\,\,  t\leq s,
	\end{array}
	\right.
	\end{equation*}
	and thus, this defines once more a $\NmuD$ with parameters $(\P;\widetilde{\alpha},\widetilde{\beta},\widetilde{\theta},\widetilde{\nu})$ for the $\gamma$-shifted system \eqref{615}. This situation motivates us to define an optimal set of constants each time that a system admits a nonuniform $\mu$-dichotomy.

	In order to provide some general concepts, let us consider an arbitrary nonautonomous linear system 
	\begin{equation}\label{601}
	\dot{z}=B(t)z,
	\end{equation} 
	where $t\mapsto B(t)$ is a locally integrable matrix valued function. In addition, let $\Psi(t,s)$ be the transition matrix associated to \eqref{601}.
	
	\begin{definition}\label{regions}
		Assume the system \eqref{601} admits $\NmuD$ with an invariant projector $\P$.   The \textbf{region of stable constants} for \eqref{601} is the set defined by
		\[
		\St_\P:=\{(\alpha,\theta)\in \mathbb{R}^2: \text{ \eqref{601} admits }\NmuD\text{ with parameters }(\P;\alpha,\beta,\theta,\nu), \text{ for some }\beta,\nu\},
		\]
		and the \textbf{region of unstable constants} for \eqref{601}  is the set defined by    \[
		\Un_\P:=\{(\beta,\nu)\in \mathbb{R}^2: \text{ \eqref{601} admits }\NmuD\text{ with parameters }(\P;\alpha,\beta,\theta,\nu), \text{ for some }\alpha,\theta\}.
		\]
	\end{definition}
	
	\begin{remark}
		{\rm Note that these regions allow for a characterization of those systems with uniform dichotomy as those for which $\St_\P\cap(\mathbb{R}\times \{0\})\neq \emptyset$ and $\Un_\P\cap(\mathbb{R}\times \{0\})\neq \emptyset$. }
	\end{remark}
	
	Now we dedicate some remarks to make an estimation of these regions.
	
	\begin{remark}\label{616} \rm 
		Clearly, when system \eqref{601} admits $\NmuD$ with parameters $(\P;\alpha,\beta,\theta,\nu)$, the regions $\St_\P$ and $\Un_\P$ are nonempty sets, and the following inclusions follow
		\[
		\St_\P\subset \{(\alpha,\theta)\in \mathbb{R}^2: \alpha<0\leq \theta;\, \alpha+\theta<0\},
		\]
		and 
		\[
		\Un_\P\subset \{(\beta,\nu)\in \mathbb{R}^2: 0\leq \nu<\beta;\,\beta-\nu>0\}.
		\]
		
		On the other hand, from the estimations of the $\NmuD$, for $t=s$, we obtain
		\[
		\|\P(s)\|\leq K\mu(s)^{\sgn(s)\theta},\quad\text{and}\quad\|\Id-\P(s)\|\leq K\mu(s)^{\sgn(s)\nu},\quad\forall\, s\in \mathbb{R},
		\]
		which means that the regions of stable and unstable constants can be estimated by the projector involved in the dichotomy. Expressly, for a projector $\P$ we can define
		\[
		\widehat{\theta}_\P=\inf\left\{\theta\geq 0:\|\P(s)\|\leq K\mu(s)^{\sgn(s)\theta},\,\forall\,s\in \mathbb{R}, \text{ for some }K>0 \right\},
		\]
		and
		\[
		\widehat{\nu}_\P=\inf\left\{\nu\geq 0:\|\Id-\P(s)\|\leq K\mu(s)^{\sgn(s)\nu},\,\forall\,s\in \mathbb{R}, \text{ for some }K>0 \right\}.
		\]
		Therefore, we can refine the above inclusions as
		\[\St_\P\subset \{(\alpha,\theta)\in \mathbb{R}^2: \alpha<0\leq \widehat{\theta}_\P\leq\theta ;\, \alpha+\theta<0\},
		\]
		and 
		\[
		\Un_\P\subset \{(\beta,\nu)\in \mathbb{R}^2: 0\leq\widehat{\nu}_\P\leq \nu<\beta;\,\beta-\nu>0\}.
		\] 
	\end{remark}
	
	However, note that this does not mean that there are some $\alpha$ and $\beta$ such that $(\alpha,\widehat{\theta}_\P)\in \St_\P$ and $(\beta,\widehat{\nu}_\P)\in \Un_\P$. Consider for example a nonuniform contraction, {\it i.e.} a system which admits $\NmuD$ but does not admit $\muD$ and the projector is the identity. In this case, clearly $\widehat{\theta}_\Id=\widehat{\nu}_\Id=0$, see e.g. \cite[Proposition 2.3]{BV}.  In other words, the regions of possible constants can be estimated by the projector, but do not depend exclusively of it, and also depend on the system itself.

	In the next, we recall te notion of $\mu$-bounded growth for the nonuniform context \cite{Silva}.

	\begin{definition}%\label{424}
		The system \eqref{601} has \textbf{nonuniform $\mu$-bounded growth} with parameter $\epsilon>0$, or just \textbf{$(\Nmu,\epsilon)$-growth}, if there are constants $\widehat{K}\geq 1$, $a\geq 0$ such that
		\[
		\|\Psi(t,s)\|\leq \widehat{K}\left(\frac{\mu(t)}{\mu(s)}\right)^{\sgn(t-s)a}\mu(s)^{\sgn(s)\epsilon},\quad\forall\,t,s\in \mathbb{R}.
		\]
		
		Moreover, if $\epsilon=0$, it is said that the system \eqref{601} has \textbf{uniform $\mu$-bounded growth} o just \textbf{$\mu$-growth}.
	\end{definition}
	
	\begin{remark}\label{617}
		{\rm     It is immediate to note that if the system \eqref{613} has $(\Nmu,\epsilon)$-growth with constants $\widehat{K}\geq 1$ and $a\geq 0$, then the $\gamma$-shifted system \eqref{615} also has $(\Nmu,\epsilon)$-growth, with constants $\widehat{K}$ and $a+|\gamma|$.
		}
	\end{remark}
	
	The following result, inspired by \cite[Lemma~1]{Castaneda6}, relates the parameters of the $\NmuD$ with the constants of the $(\Nmu,\epsilon)$-growth of a system.
	\begin{lemma}\label{612}
		Assume the system \eqref{601} has $(\Nmu,\epsilon)$-growth with constants $a\geq 0$, $\widehat{K}\geq 1$, and also admits $\NmuD$ with parameters $(\P;\alpha,\beta,\theta,\nu)$. For $\P\neq 0$, we have $-(a+\epsilon)\leq\alpha$. In particular
		\[
		\St_\P\subset \{(\alpha,\theta)\in \mathbb{R}^2: -(a+\epsilon)\leq\alpha\}.
		\]
		
		Analogously, for $\P\neq \Id$, we have $\beta\leq a+\epsilon$. In particular
		\[
		\Un_\P\subset \{(\beta,\nu)\in \mathbb{R}^2: \beta\leq a+\epsilon\}.
		\]
	\end{lemma}
	
	\begin{proof}
		We prove only the first inequality, since the other case follows similarly. Let $0<s<t$. Since $\P\neq 0$, we can chose $\xi\in \P(s)\mathbb{R}^d$ with $\xi\neq 0$. Then we have the following estimation
		\begin{align*}
		\|\xi\|=\|\Psi(s,t)\Psi(t,s)\xi\|&\leq \|\Psi(s,t)\|\cdot \|\Psi(t,s)\P(s)\xi\|\\
		&\leq\widehat{K}\left(\frac{\mu(s)}{\mu(t)}\right)^{\sgn(s-t)a}\mu(t)^{\sgn(t)\epsilon}K\left(\frac{\mu(t)}{\mu(s)}\right)^{\alpha}\mu(s)^{\sgn(s)\theta}\|\xi\|,
		\end{align*}
		from which we deduce
		\[
		1\leq \widehat{K}K\mu(t)^{a+\alpha+\epsilon}\mu(s)^{\theta-a-\alpha}, \quad \forall\,t>s.
		\]
		
		Note that considering $a+\alpha+\epsilon<0$ and taking the limit as $t\to +\infty$, the right-hand side tends to zero, which is clearly a contradiction. Thus, we must have $-(a+\epsilon)\leq \alpha$.
	\end{proof}
	
	\begin{remark}\label{619}
		{ \rm Taking into account Remark~\ref{616} and Lemma~\ref{612}, if we assume the system \eqref{601} has $(\Nmu,\epsilon)$-growth with constants $a\geq 0$ and $\widehat{K}\geq 1$, and also admits $\NmuD$ with invariant projector $\P$, then we deduce 
			\[
			\St_\P\subset \{(\alpha,\theta)\in \mathbb{R}^2: -(a+\epsilon)\leq\alpha<0\leq \widehat{\theta}_\P\leq\theta ;\, \alpha+\theta<0\},
			\]
			and 
			\[
			\Un_\P\subset \{(\beta,\nu)\in \mathbb{R}^2: 0\leq\widehat{\nu}_\P\leq \nu<\beta\leq a+\epsilon;\,\beta-\nu>0\}.
			\]
			
			Therefore, both $\St_\P$ and $\Un_\P$ are contained in bounded sets, as illustrated on Figure \ref{635}.  
		}
	\end{remark}

	\begin{figure}[h]
		\centering
		\subfloat[Region $\St_\P$]{
			%\label{f:gato}
			\includegraphics[width=0.48\textwidth]{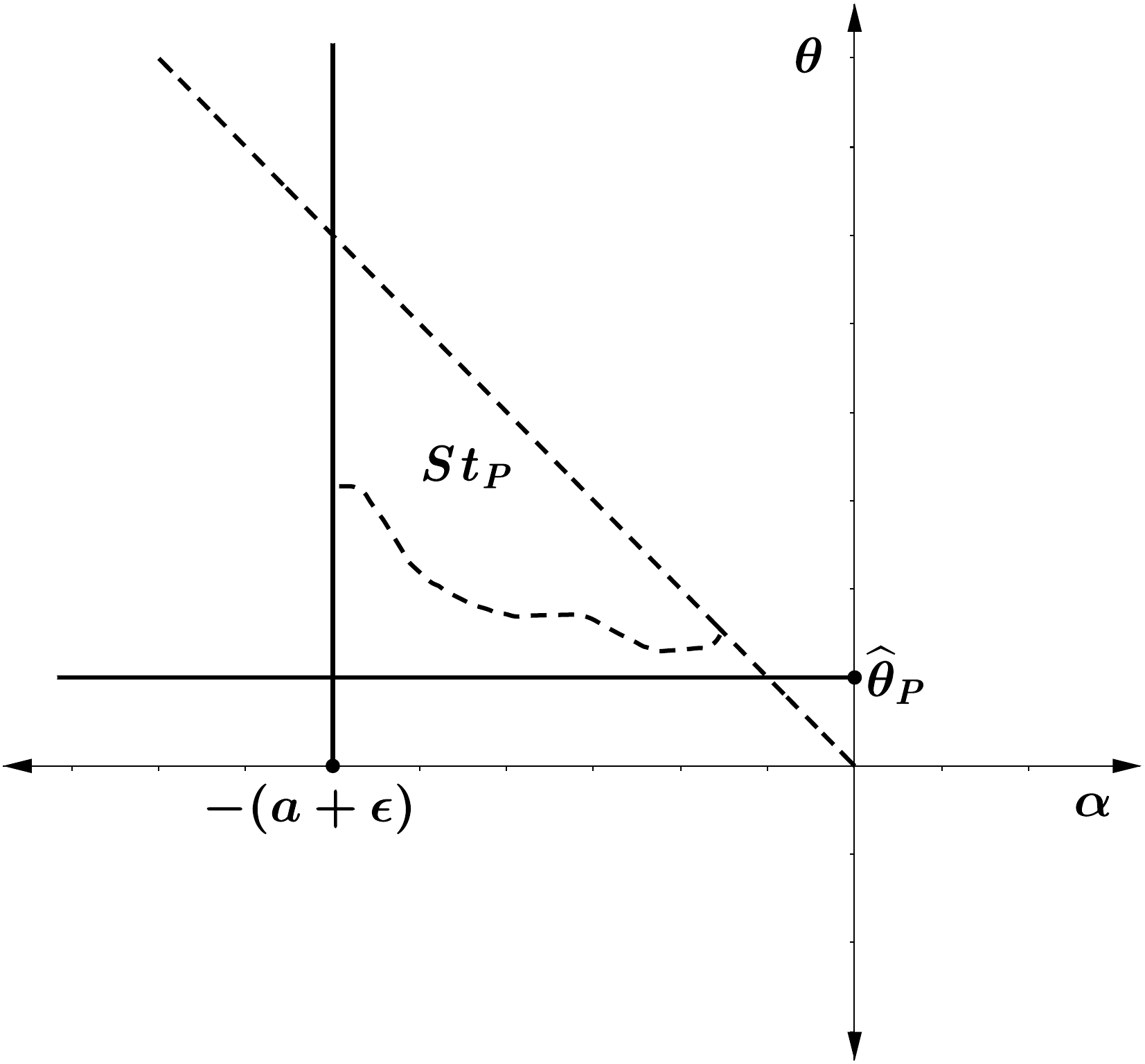}}
		\subfloat[Region $\Un_\P$]{
			% \label{f:tigre}
			\includegraphics[width=0.48\textwidth]{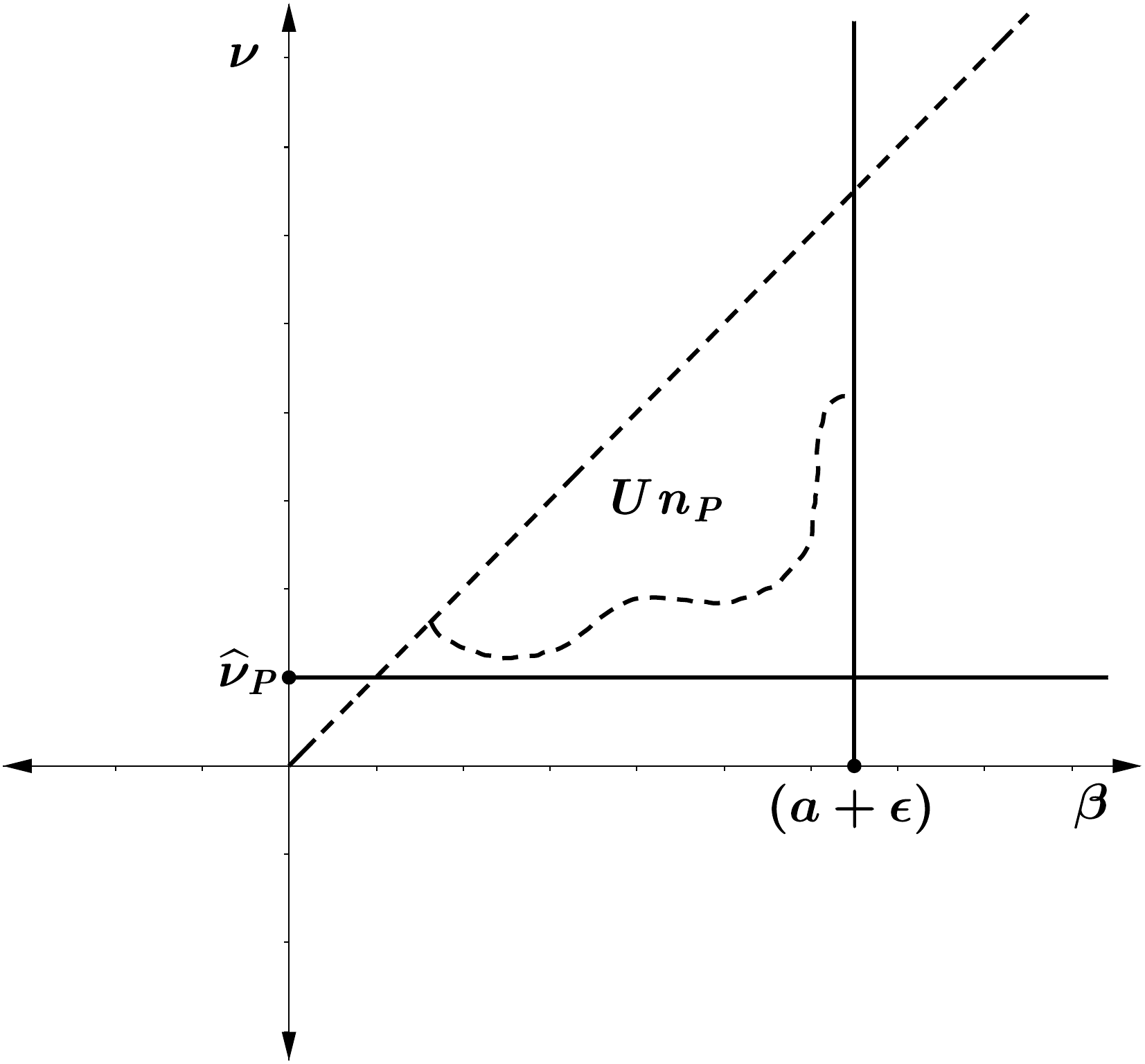}}
		\caption{Contentions from Remark \ref{619}}
		\label{635}
	\end{figure}
	
	The preceding remark motivates us  to introduce the concept of optimal parameters associated with the $\NmuD$.
	
	\begin{definition}\label{stpunp}
		Assume the system \eqref{601} admits $(\Nmu,\epsilon)$-growth and $\NmuD$ with an invariant projector $\P$. 
		\begin{itemize}
			\item The \textbf{optimal stable ratio} is defined by
			\[
			\st_\P:=\inf\{\alpha+\theta:(\alpha,\theta)\in \St_\P\}.
			\]
			
			\item The \textbf{optimal unstable ratio} is defined by
			\[
			\un_\P:=\sup\{\beta-\nu:(\beta,\nu)\in \Un_\P\}.
			\]
		\end{itemize}
	\end{definition}

	In the next, we recall a characterization of the nonuniform $\mu$-dichotomy spectrum 
	
	\begin{theorem} \label{602}
		{\rm (\cite[Theorem 8]{Silva})} 
		Assume the system \eqref{613} has $(\Nmu,\epsilon)$-growth. Then there exists some $n\in \{1,\dots,d\}$ such that the nonuniform $\mu$-dichotomy spectrum of \eqref{613} is nonempty, compact and has the form
		\[
		\Si_\NmuD(A)=\bigcup_{i=1}^n[a_i,b_i],
		\]
		where $a_i\leq b_i<a_{i+1}$.
	\end{theorem}

	The compact intervals $[a_i,b_i]$, for $i=1,\dots,n$, described on Theorem \ref{602}, are called spectral intervals. Each open interval that compose the resolvent set is termed a {\it spectral gap}. Let us denote each spectral gap by $(b_i,a_{i+1})$, for $i=0,\dots,n+1$, where $b_0=-\infty$ and $a_{n+1}=+\infty$. 
	
	\begin{remark}
		{\rm  
			For any element in the resolvent set, the correspondent shifted system admits $\NmuD$. In \cite[Lemma 5, Lemma 6]{Silva} it was established that for any two numbers $\gamma,\zeta$ contained in the same spectral gap $(b_i,a_{i+1})$, the projector associated to the dichotomy of both the $\gamma$-shifted and the $\zeta$-shifted systems are the same. In other words, every spectral gap has a unique invariant projector $\P$ associated to it.

			It is immediately deduced from the construction of the spectrum \cite{Silva}, that on every bounded spectral gap, the associated projector $\P$ is neither $0$ nor $\Id$. Nevertheless, on the unbounded spectral gap $(-\infty,a_1)$, if it is contained on the resolvent set, the projector is always $0$. Similarly, if the spectral gap $(b_n,+\infty)$ is contained on the resolvent set, its associated projector is always $\Id$. Note that by Theorem \ref{602}, if a system has $(\Nmu,\epsilon)$-growth, then it is guaranteed that both unbounded spectral gaps are contained in its resolvent set.
		}
	\end{remark}

	Let $\gamma$ be an element in the spectral gap $(b_i,a_{i+1})$. Since the $\gamma$-shifted system admits $\NmuD$, it follows that at least one of the regions of constants is nonempty. Specifically, for $i=1,\dots,n$, the region of stable constants is nonempty and we denote it by $\St_\P^\gamma$, while for $i=0,\dots,n-1$, the region on unstable constants is nonempty and we denote it by $\Un_\P^\gamma$. Hence, for the correspondent index $i$, we can define two functions, $\St_\P\colon(b_i,a_{i+1})\to 2^{\mathbb{R}}$, $\gamma\mapsto \St_\P^\gamma$ and $\Un_\P\colon(b_i,a_{i+1})\to 2^{\mathbb{R}}$, $\gamma\mapsto \Un_\P^\gamma$. 
	
	From Remarks~\ref{617} and \ref{619}, we deduce
	\[
	\emptyset\neq\St_\P^\gamma\subset \left\{(\alpha,\theta)\in \mathbb{R}^2: -(a+\epsilon+ |\gamma|)\leq \alpha<0\leq \widehat{\theta}_\P\leq\theta;\, \alpha+\theta<0  \right\},
	\]
	for $i=1,\dots,n$, and 
	\[
	\emptyset\neq\Un_\P^\gamma\subset \left\{(\beta,\nu)\in \mathbb{R}^2: 0\leq\widehat{\nu}_\P\leq \nu<\beta\leq a+\epsilon+ |\gamma|;\,\beta-\nu>0 \right\},
	\]
	for $i=0,\dots,n-1$. Note that in both cases, the sets on the right-hand side of the previous inclusions are bounded. Moreover, in the case that the spectral gap $(b_i,a_{i+1})$ is bounded, {\it i.e.} for $i=1,\dots,n-1$, the sets on the right-hand side are uniformly bounded, that is
	\[
	\emptyset\neq\St_\P^\gamma\subset \left\{(\alpha,\theta)\in \mathbb{R}^2: -(a+\epsilon+ \max\{|b_i|,|a_{i+1}|\})\leq \alpha<0\leq \widehat{\theta}_\P\leq\theta;\, \alpha+\theta<0  \right\},
	\]
	and
	\[
	\emptyset\neq\Un_\P^\gamma\subset \left\{(\beta,\nu)\in \mathbb{R}^2: 0\leq\widehat{\nu}_\P\leq \nu<\beta\leq a+\epsilon+ \max\{|b_i|,|a_{i+1}|\};\,\beta-\nu>0 \right\},
	\]
	therefore, the functions $\gamma \mapsto\St_\P^\gamma$ and $\gamma\mapsto\Un_\P^\gamma$ will be bounded on every bounded spectral gap.

	Similarly, in view of Definition~\ref{stpunp}, it makes sense to define the following maps
	\begin{definition}\label{optimalmaps}
		Assume the system \eqref{601} admits $(\Nmu,\epsilon)$-growth.    Let $(b_i,a_{i+1})$ be an spectral gap. 
		\begin{itemize}
			\item For $i=1,\dots,n$, the function $\st_\P\colon(b_i,a_{i+1})\to \mathbb{R}$ defined as 
			\[
			\gamma\mapsto \st_\P^\gamma:=\inf\{\alpha+\theta:(\alpha,\theta)\in \St_\P^{\gamma}\},
			\]
			is called \textbf{function of optimal stable ratio}.
			
			\item For $i=0,\dots,n-1$, the function $\un_\P\colon(b_i,a_{i+1})\to \mathbb{R}$ defined as 
			\[
			\gamma\mapsto \un_\P^\gamma:=\sup\{\beta-\nu:(\beta,\nu)\in \Un_\P^{\gamma}\},
			\]
			is called \textbf{function of optimal unstable ratio}.
		\end{itemize}
	\end{definition}
	
	In the remainder of this section, we will delve into the properties of the functions $\st_\P$ and $\un_\P$. 
	
	\begin{lemma}\label{622}
		On every spectral gap, the maps $\st_\P$ and $\un_\P$ are decreasing.
	\end{lemma}
	\begin{proof}
		We prove that $\st_\P$ is decreasing and $\un_\P$ follows similarly.   Let $\gamma,\zeta\in (b_i,a_{i+1})$ with $\gamma<\zeta$. We will establish that $\St_\P^\gamma\subset\St_\P^\zeta$, from which it follows that $\st_\P^\zeta\leq\st_\P^\gamma$. For $(\alpha,\theta)\in \St_\P^\gamma$, there exists a constant $K>0$ such that
		\begin{equation*} %\label{610}
		\|\Phi_\gamma(t,s)\P(s)\|\leq K\left(\frac{\mu(t)}{\mu(s)}\right)^{\alpha}\mu(s)^{\sgn(s)\theta},\quad \forall\,t\geq s.
		\end{equation*}
		
		Now, as $\zeta> \gamma$, then for $t\geq s$ we have $\left(\frac{\mu(t)}{\mu(s)}\right)^{\gamma-\zeta}<1$. Moreover, we have the identity $\Phi_\zeta(t,s)=\Phi_\gamma(t,s)\left(\frac{\mu(t)}{\mu(s)}\right)^{\gamma-\zeta}$, for all $t,s\in\mathbb{R}$. Thus, we get the estimation
		\begin{align*}
		\|\Phi_\zeta(t,s)\P(s)\|&\leq K\left(\frac{\mu(t)}{\mu(s)}\right)^{\alpha+\gamma-\zeta}\mu(s)^{\sgn(s)\theta}\\
		&\leq K\left(\frac{\mu(t)}{\mu(s)}\right)^{\alpha}\mu(s)^{\sgn(s)\theta} ,\quad \forall\,t\geq s,
		\end{align*}
		which implies that $(\alpha,\theta)\in \St_\P^\zeta$. Therefore, we infer that $\St_\P^\gamma\subset \St_\P^\zeta$, obtaining the desired result. 
	\end{proof} 
	
	\begin{proposition}\label{625}
		On every spectral gap, the maps $\st_\P$ and $\un_\P$ are continuous.
	\end{proposition}
	
	\begin{proof}
		
		We will prove $\st_\P$ is continuous and $\un_\P$ will follow analogously.

		As $\st_\P$ is a decreasing function, any discontinuity is a jump discontinuity. Let $\gamma\in (b_i,a_{i+1})$  and suppose $\st_\P$ is not left-continuous at $\gamma$. Then, there is some $\epsilon_1>0$ such that $\st_\P^\zeta>\st_\P^\gamma+\epsilon_1$ for every $\zeta<\gamma$. Chose $\epsilon_2>0$ with $\epsilon_2<\epsilon_1/3$. Chose now $(\alpha^\gamma,\theta^\gamma)\in \St_\P^\gamma$ such that $\st_\P^\gamma<\alpha^\gamma+\theta^\gamma<\st_\P^\gamma+\epsilon_2$. Then, there exists $K>0$ such that
		\begin{equation*} %\label{610}
		\|\Phi_\gamma(t,s)\P(s)\|\leq K\left(\frac{\mu(t)}{\mu(s)}\right)^{\alpha^\gamma}\mu(s)^{\sgn(s)\theta^\gamma},\quad \forall\,t\geq s.
		\end{equation*}

		Chose now $\epsilon_3>0$ with $\epsilon_3<\min\{-(\alpha^\gamma+\theta^\gamma),\epsilon_1/3\}$. For $\zeta\in (\gamma-\epsilon_3,\gamma)$ we have
		\begin{align*}
		\|\Phi_\zeta(t,s)\P(s)\|&\leq K\left(\frac{\mu(t)}{\mu(s)}\right)^{\alpha^\gamma+\gamma-\zeta}\mu(s)^{\sgn(s)\theta^\gamma}\\
		&\leq K\left(\frac{\mu(t)}{\mu(s)}\right)^{\alpha^\gamma+\epsilon_3}\mu(s)^{\sgn(s)\theta^\gamma} ,\quad \forall\,t\geq s.
		\end{align*}
		In addition, since $\alpha^\gamma+\epsilon_3+\theta^\gamma<0$, we infer that  $(\alpha^\gamma+\epsilon_3,\theta^\gamma)\in \St_\P^\zeta$. However, we also have the estimation
		\[\st_\P^\zeta\leq \alpha^\gamma+\epsilon_3+\theta^\gamma<\st_\P^\gamma+\epsilon_3+\epsilon_2<\st_\P^\gamma+\epsilon_1<\st_\P^\zeta,\] 
		which is evidently contradictory. Thus, $\st_\P$ must be left-continuous.

		Suppose now that $\st_\P$ is not right-continuous at $\gamma$. Then, there is some $\epsilon_4>0$ such that $\st_\P^\zeta<\st_\P^\gamma-\epsilon_4$ for every $\zeta>\gamma$. Chose $\epsilon_5>0$ with $\epsilon_5<\epsilon_4/3$, and $\zeta>\gamma$ with $\zeta-\gamma<\epsilon_5$. Chose now $\epsilon_6>0$ with $\epsilon_6<\epsilon_4/3$, and a pair  $(\alpha^\zeta,\theta^\zeta)\in \St_\P^\zeta$ such that $\st_\P^\zeta<\alpha^\zeta+\theta^\zeta<\st_\P^\zeta+\epsilon_6$. This implies that there exists $K>0$ such that
		\begin{equation*} %\label{610}
		\|\Phi_\zeta(t,s)\P(s)\|\leq K\left(\frac{\mu(t)}{\mu(s)}\right)^{\alpha^\zeta}\mu(s)^{\sgn(s)\theta^\zeta},\quad \forall\,t\geq s.
		\end{equation*}
		This leads us to deduce
		\begin{align}\label{624}
		\|\Phi_\gamma(t,s)\P(s)\|&\leq K\left(\frac{\mu(t)}{\mu(s)}\right)^{\alpha^\zeta+\zeta-\gamma}\mu(s)^{\sgn(s)\theta^\zeta}\nonumber\\
		&\leq K\left(\frac{\mu(t)}{\mu(s)}\right)^{\alpha^\zeta+\epsilon_5}\mu(s)^{\sgn(s)\theta^\zeta} ,\quad \forall\,t\geq s.
		\end{align}
		
		Note that
		\[\alpha^\zeta+\theta^\zeta+\epsilon_5<\st_\P^\zeta+\epsilon_6+\epsilon_5<\st_\P^\gamma-\epsilon_4+\epsilon_6+\epsilon_5<\st_\P^\gamma\leq 0.\]
		
		First, this implies $\alpha^\zeta+\theta^\zeta+\epsilon_5<0$, hence \eqref{624} implies $(\alpha^\zeta+\epsilon_5,\theta^\zeta)\in \St_\P^\gamma$. Nevertheless, we also have $\alpha^\zeta+\theta^\zeta+\epsilon_5<\st_\P^\gamma$, which contradicts the minimality of $\st_\P^\gamma$. Thus, $\st_\P$ must be right continuous.
	\end{proof}

	The previous results imply that both $\st_\P$ and $\un_\P$ have lateral limits on every bounded spectral gap. Similarly, in the case of the spectral gaps $(-\infty,a_1)$ and $(b_n,+\infty)$, although the maps might be unbounded, the lateral limits towards $a_1$ and $b_n$ still exist.

	Next, we present the principal property of functions $\st_\P$ and $\un_\P$, which will be crucial to establish the novelty and the main result of the subsequent section.

	\begin{theorem}\label{620}
		Assume the system \eqref{601} admits $(\Nmu,\epsilon)$-growth.    For every bounded spectral gap $(b_i,a_{i+1})$, we have
		$$\lim_{\gamma\to b_i^+}\st_\P^\gamma=\lim_{\gamma\to a_{i+1}^-}\un_\P^\gamma=0.$$
		
		Moreover, on the unbounded spectral gaps we have $\displaystyle\lim_{\gamma\to b_n^+}\st_\Id^\gamma=\lim_{\gamma\to a_1^-}\un_0^\gamma=0$.
	\end{theorem}
	
	\begin{proof}
		We will prove $\lim_{\gamma \to b_i^+}\st_\P^\gamma=0$ for a bounded spectral gap (the other limits follow similarly). Since the map $\st_\P$ is decreasing, negative and continuous, it follows that the lateral limits exist and are non-positive .

		Suppose $\lim_{\gamma\to b_i^+}\st_\P^\gamma=M<0$. We chose $\epsilon_1>0$ such that $M+\epsilon_1<0$. Moreover, for each $\gamma\in (b_i,a_{i+1})$ we chose $(\alpha^\gamma,\theta^\gamma)\in \St_\P^\gamma$ such that
		\[\st_\P^\gamma<\alpha^\gamma+\theta^\gamma<\st_\P^\gamma+\epsilon_1.\]
		
		Note that this defines a new function $\widetilde{\st}_\P:(b_i,a_{i+1})\to \mathbb{R}$, $\gamma\mapsto \widetilde{\st}_\P^\gamma:=\alpha^\gamma+\theta^\gamma$ which verifies $\st_\P^\gamma<\widetilde{\st}_\P^\gamma<\st_\P^\gamma+\epsilon_1$. Then
		\[\liminf_{\gamma\to b_i^+}\widetilde{\st}_\P^\gamma:=\widetilde{M}\leq M+\epsilon_1<0.\]
		
		Let us consider $\epsilon_2\in (0,-\widetilde{M}/4)$. The previous lower limit implies that there is $\delta_1=\delta_1(\epsilon_2)>0$ such that for every $\delta\in (0,\delta_1)$ there is some $\gamma\in (b_i,b_i+\delta)$ such that $\widetilde{\st}_\P^\gamma\in (\widetilde{M}-\epsilon_2,\widetilde{M}+\epsilon_2)$.

		Chose now $\delta_2=\min\{\delta_1(\epsilon_2),-\widetilde{M}/4\}$. Then, there is some $\zeta\in(b_i,b_i+\delta_2)$ such that $\widetilde{\st}_\P^\zeta\in (\widetilde{M}-\epsilon_2,\widetilde{M}+\epsilon_2)$.

		As by definition $(\alpha^\zeta,\theta^\zeta)\in \St_\P^\zeta$, then the $\zeta$-shifted system admits $\NmuD$ with parameters $(\P;\alpha^\zeta,\beta,\theta^\zeta,\nu)$, for some $\beta$ and $\nu$. In other words
		\begin{equation*}
		\left\{ \begin{array}{lc}
		\|\Phi_\zeta(t,s)\P(s)\|\leq K
		\left( \mathlarger{\frac{\mu(t)}{\mu(s)}}\right)^{\alpha^\zeta}\mu(s)^{\sgn(s)\theta^\zeta}, &\forall\,\, t \geq s \\
		\\ \|\Phi_\zeta(t,s)[\Id-\P(s)]\|\leq K
		\left( \mathlarger{\frac{\mu(t)}{\mu(s)}}\right)^{\beta}\mu(s)^{\sgn(s)\nu}, &\forall\,\,  t\leq s.
		\end{array}
		\right.
		\end{equation*}
		
		Moreover, we have the identity
		\[\Phi_{b_i}(t,s)=\Phi_\zeta(t,s)\left(\frac{\mu(t)}{\mu(s)}\right)^{\zeta-b_i},\qquad \text{ for all $t,s\in\mathbb{R}$},
		\]
		from which we obtain
		\begin{equation*}
		\left\{ \begin{array}{lc}
		\|\Phi_{b_i}(t,s)\P(s)\|\leq K
		\left( \mathlarger{\frac{\mu(t)}{\mu(s)}}\right)^{\alpha^\zeta+(\zeta-b_i)}\mu(s)^{\sgn(s)\theta^\zeta}, &\forall\,\, t \geq s \\
		\\ \|\Phi_{b_i}(t,s)[\Id-\P(s)]\|\leq K
		\left( \mathlarger{\frac{\mu(t)}{\mu(s)}}\right)^{\beta+(\zeta-b_i)}\mu(s)^{\sgn(s)\nu}, &\forall\,\,  t\leq s,
		\end{array}
		\right.
		\end{equation*}
		where clearly $\beta+(\zeta-b_i)-\nu>\beta-\nu>0$, and 
		\[
		\alpha^\zeta+(\zeta-b_i)+\theta^\zeta<\widetilde{\st}_\P^\zeta+\delta_2<\widetilde{M}/2<0.
		\]
		Therefore, the $b_i$-shifted system admits $\NmuD$. This is a contradiction due to $b_i$ belongs to $\Sigma_\NmuD(A)$. In consequence, we infer that $M=0$. 
	\end{proof}
	
	\begin{remark}\label{639}
		{\rm
			Since the function $\gamma \mapsto \st_\P^\gamma$ is decreasing, negative and continuous on every spectral gap, if there is some number $b$ such that $\lim_{\gamma\to b^+}\st_\P^\gamma=0$, then $b$ must be a boundary point of said spectral gap, {\it i.e.} $b$ is exactly one of the $b_i$, with $i=1,\dots,n$, of the spectrum, as tailored on Theorem \ref{602}. Similarly, if $\lim_{\gamma\to a^-}\un_\P^\gamma=0$, for some $a$, then $a$ must be a boundary point in the spectrum. Observe that this assertion can be regarded as the reciprocal to the previous theorem.    
		}
	\end{remark}
	
	\begin{corollary}\label{621}
		For every $\varepsilon>0$, there exists $\delta>0$ such that 
		\[
		\gamma\in (b_i,b_i+\delta) \Longrightarrow |\st_\P^\gamma|<\varepsilon,
		\]
		and 
		\[
		\gamma\in (a_{i}-\delta,a_{i}) \Longrightarrow \un_\P^\gamma<\varepsilon, 
		\]
		for all $i=1,\dots,n$.
	\end{corollary}
	
	\begin{remark}
		{\rm
			For systems that present uniform $\mu$-dichotomy, the previous result implies that the constants $\alpha^\gamma$ or $\beta^\gamma$ tend to zero when $\gamma$ approaches the spectrum. 
		}
	\end{remark}

	To conclude this section, we address a property of the optimal ratio maps on the unbounded spectral gaps.
	
	\begin{lemma}\label{629}
		Assume the system \eqref{613} has $(\Nmu,\epsilon)$-growth with constants $\widehat{K}\geq 1$ and $a\geq 0$. Then
		\[\lim_{\gamma\to -\infty}\un_0^\gamma=-\lim_{\gamma\to +\infty}\st_\Id^\gamma=+\infty.\]
		
	\end{lemma}
	
	\begin{proof}
		We will prove $\lim_{\gamma\to +\infty}\st_\Id^\gamma=-\infty$ and the other limit will follow similarly. As $\st_\Id$ is decreasing and continuous, it is enough to prove that the map cannot be bounded. Suppose that $\st_\Id$ is bounded, {\it i.e.} there is some $N<0$ such that $N<\st_\Id^\gamma<0$ for every $\gamma>b_n$. Consider $t\geq s$. By the assumption of nonuniform $\mu$-bounded growth, we have
		\[
		\|\Phi(t,s)\|\leq\widehat{K}\left(\frac{\mu(t)}{\mu(s)}\right)^{a}\mu(s)^{\sgn(s)\epsilon}.
		\]
		Thus, for $\gamma$ large enough, the estimation
		\[\|\Phi_\gamma(t,s)\|\leq\widehat{K}\left(\frac{\mu(t)}{\mu(s)}\right)^{a-\gamma}\mu(s)^{\sgn(s)\epsilon},\]
		defines a dichotomy with parameters $(\Id;a-\gamma,*,\epsilon,*)$. Moreover, for $\gamma>a+\epsilon-N$, we have $\st_\Id^\gamma\leq a-\gamma+\epsilon<N$, which is a contradiction.
	\end{proof}

	\section{Kinematic similarity and spectra noninvariance}
	
	In this section, we aim to clarify certain inconsistencies that have emerged from earlier studies regarding the preservation of the nonuniform $\mu$-dichotomy spectra for systems which are nonuniformly kinematically similar.
	
	Let us commence by recalling the concept of kinematic similarity frequently used in the nonuniform context, see e.g. \cite{Chu,Huerta,Silva,Xiang}.
	
	\begin{definition}\label{Kinsim}
		Let $\mu\colon\mathbb{R}\to(0,+\infty)$ be a growth rate and let $\varepsilon\geq0$ be given. The systems \eqref{613} and \eqref{601} are \textbf{nonuniformly $(\mu,\varepsilon)$-kinematically similar}, if there is a constant $M_\varepsilon>0$, and a differentiable matrix valued function $S\colon\mathbb{R}\to GL_d(\mathbb{R})$ satisfying the following properties:
		\begin{itemize}
			\item[(i)] $\|S(t)\|\leq M_\varepsilon\mu(t)^{\sgn(t)\varepsilon}$, for all $t\in\mathbb{R}$.
			
			\item[(ii)]$\|S(t)^{-1}\|\leq M_\varepsilon\mu(t)^{\sgn(t)\varepsilon}$, for all $t\in\mathbb{R}$.
			
			\item[(iii)]If $t\mapsto y(t)$ is a solution for \eqref{613}, then $t\mapsto z(t):=S(t)^{-1}y(t)$ is a solution for \eqref{601}.
			
			\item[(iv)] If $t\mapsto z(t)$ is a solution of \eqref{601}, then $t\mapsto y(t):=S(t)z(t)$ is a solution of \eqref{613}.
			
		\end{itemize}

		In the particular case that $\varepsilon=0$, it is said that systems \eqref{613} and \eqref{601}  are \textbf{uniformly kinematically similar}.
		
	\end{definition}
	
	\begin{remark}\rm 
		Every matrix valued function $S\colon\mathbb{R}\to GL_d(\mathbb{R})$ satisfying (i) and (ii) for some $\varepsilon>0$ is called a \emph{nonuniform Lyapunov matrix function} with respect to $\mu$, and the change of variables $y(t)=S(t)z(t)$ is called a \emph{nonuniform Lyapunov transformation} with respect to $\mu$. As usual, in the case $\varepsilon=0$ the term ``nonuniform" is replaced with ``uniform". 
	\end{remark}

	The next result, which is a nonuniform version borrowed from \cite[Lemma 2.1]{Siegmund2}, characterizes the nonuniform $(\mu,\varepsilon)$-kinematic similarity through special cases of nonuniform Lyapunov matrix functions.
	
	\begin{lemma}\label{604}
		{\rm (\cite[Lemma 13]{Silva})} Let $S(t)$ be a nonuniform Lyapunov matrix function with respect to $\mu$. Then, the following statements are equivalent:
		\begin{itemize}
			\item [(a)] $S(t)$ verifies {\rm(iii)} and {\rm (iv)} from Definition \ref{Kinsim}.
			\item [(b)] The identity $\Phi(t,s)S(s)=S(t)\Psi(t,s)$ holds, for all $t,s\in \mathbb{R}$, where $\Phi$ and $\Psi$ are the evolution operators of \eqref{613} and \eqref{601}, respectively.
			\item [(c)] $S(t)$ is a solution of $S'=A(t)S-SB(t)$.
		\end{itemize}
	\end{lemma}
	
	The following result is a simple but useful consequence of the previous lemma.
	
	\begin{corollary}\label{623}
		Assume the systems \eqref{613} and \eqref{601} are nonuniformly $(\mu,\varepsilon)$-kinematically similar. Then for every $\gamma\in \mathbb{R}$, the $\gamma$-shifted systems from both \eqref{613} and \eqref{601} are nonuniformly $(\mu,\varepsilon)$-kinematically similar with the same Lyapunov transformation and the same $\varepsilon$. 
	\end{corollary}
	
	\begin{proof}
		This follows from the second equivalence on Lemma \ref{604}, as both $\Phi_\gamma$ and $\Psi_\gamma$ are obtained by multiplying $\Phi$ and $\Psi$ by the same scalar evolution.
	\end{proof}
	
	In the next, we slightly reformulate a result established in \cite[Lemma 14]{Silva}. We include some extra details related with the parameters of the $\NmuD$.
	
	\begin{lemma}\label{605}
		Assume the system \eqref{613} admits $\NmuD$ with parameters $(\P;\alpha,\beta,\theta,\nu)$, where $\P$ is neither $0$ nor $\Id$, and is nonuniformly $(\mu,\varepsilon)$-kinematically similar to the system \eqref{601} with a Lyapunov function $S$, and
		\begin{equation}\label{603}
		\varepsilon< \frac{1}{3}\min\{-(\alpha+\theta),\beta-\nu\}.  
		\end{equation}
		
		Then, the system \eqref{601} admits $\NmuD$ with parameters  $(\Q;\alpha+\varepsilon,\beta-\varepsilon,\theta+2\varepsilon,\nu+2\varepsilon)$, where $Q$ is the invariant projector for \eqref{601} defined by $\Q(s)=S(s)^{-1}\P(s)S(s)$, for all $s\in\mathbb{R}$.
	\end{lemma}
	
	\begin{proof}
		We have
		\begin{equation*} 
		\left\{ \begin{array}{lc}
		\|\Phi(t,s)\P(s)\|\leq K
		\left( \mathlarger{\frac{\mu(t)}{\mu(s)}}\right)^\alpha\mu(s)^{\sgn(s)\theta}, &\forall\,\, t \geq s, \\
		\\ \|\Phi(t,s)[\Id-\P(s)]\|\leq K
		\left( \mathlarger{\frac{\mu(t)}{\mu(s)}}\right)^\beta\mu(s)^{\sgn(s)\nu}, &\forall\,\,  t\leq s.
		\end{array}
		\right.
		\end{equation*}
		
		From the identity $\Psi(t,s)=S(t)^{-1}\Phi(t,s)S(s)$, we obtain
		\begin{align*}
		\|\Psi(t,s)\Q(s)\|&=\|S(t)^{-1}\Phi(t,s)\P(s)S(s)\|\\
		&\leq KM_\varepsilon^2 \left(\frac{\mu(t)}{\mu(s)}\right)^{\alpha+\varepsilon}\mu(s)^{\sgn(s)(\theta+2\varepsilon)}, \qquad \forall\,t\geq s,
		\end{align*} 
		and
		\begin{align*}
		\|\Psi(t,s)[\Id-\Q(s)]\|&=\|S(t)^{-1}\Phi(t,s)[\Id-\P(s)]S(s)\|\\
		&\leq KM_\varepsilon^2 \left(\frac{\mu(t)}{\mu(s)}\right)^{\beta-\varepsilon}\mu(s)^{\sgn(s)(\nu+2\varepsilon)}, \qquad \forall\,t\leq s.
		\end{align*}
		
		Therefore, as \eqref{603} is verified, the system \eqref{601} admits $\NmuD$ with parameters $(\Q;\alpha+\varepsilon,\beta-\varepsilon,\theta+2\varepsilon,\nu+2\varepsilon)$.
	\end{proof}
	
	\begin{remark}
		{\rm
			The preceding lemma remains valid when the projector is either $0$ or $\Id$. In such cases, we only have to consider the constants $\beta$ and $\nu$ or $\alpha$ and $\theta$, respectively, due to the absence of stable or unstable manifolds.}
	\end{remark}

	In the statement of \cite[Lemma 14]{Silva} is asserted that under assumptions of nonuniform $(\mu,\varepsilon)$-kinematic similarity for systems \eqref{613} and \eqref{601}, with $\varepsilon$ satisfying \eqref{603}, the nonuniform $\mu$-dichotomy spectrum is invariant, {\it i.e.}  $\Si_{\NmuD}(A)=\Si_{\NmuD}(B)$, or equivalently, $\rho_{\NmuD}(A)=\rho_{\NmuD}(B)$. However, the origin of the parameters $\alpha$, $\beta$, $\theta$ and $\nu$ involved in \eqref{603} is not explicitly clarified on \cite[Lemma 14]{Silva}.

	\begin{remark}\label{660}
		{\rm Let us explain the dilemma with the preservation of the nonuniform $\mu$-dichotomy spectrum for systems \eqref{613} and \eqref{601} which are nonuniformly $(\mu,\varepsilon)$-kinematically similar: note that, from Corollary~\ref{623}, the corresponding $\gamma$-shifted systems associated to \eqref{613} and \eqref{601} are nonuniformly $(\mu,\varepsilon)$-kinematically similar. Hence, for  $\gamma\in\rho_{\NmuD}(A)$, and in light of Lemma~\ref{605}, it becomes evident that in the nonuniform case, the fulfillment of condition \eqref{603} is imperative for establishing the $\NmuD$ of the $\gamma$-shifted system \eqref{601}, and this dichotomy holds precisely with parameters $(\Q;\alpha+\varepsilon,\beta-\varepsilon,\theta+2\varepsilon,\nu+2\varepsilon)$. Nevertheless, the parameters in \eqref{603} strictly depend of the election of $\gamma\in\mathbb{R}$, {\it i.e.} for a different $\gamma$ in $\rho_{\NmuD}(A)$ we will have distinct constants $\alpha, \beta,\theta$ and $\nu$. Moreover, in sight of Theorem \ref{620}, for every $\varepsilon>0$, if $\gamma$ is close enough to the spectrum, we obtain either $\st_\P^\gamma+3\varepsilon>0$ or $\un_\P^\gamma-3\varepsilon<0$. Consequently, for every $(\alpha,\theta)\in \St_\P^\gamma$ or $(\beta,\nu)\in \Un_\P^\gamma$ we would obtain $\alpha+\varepsilon+\theta+2\varepsilon>0$ or $\beta-\nu-3\varepsilon<0$, which implies $(\alpha+\varepsilon,\theta+2\varepsilon)\not\in \St_\Q^\gamma$ or $(\beta-\varepsilon,\nu+2\varepsilon)\not\in \Un_\Q^\gamma$, independently of the projector $\Q$. In other words, for elements $\gamma$ which are close enough to the spectrum of the first system, the previous lemma does not allow for a proof that such element lies on the resolvent set of the second system.}
		
	\end{remark}
	
	Considering the aforementioned explanation, \cite[Lemma 14]{Silva} is apparently contradictory in light of the following examples.
	
	\begin{example}\label{ex1}
		{\rm
			Consider the systems $\dot{x}=-10x(t)$ and $\dot{y}=-12y(t)$.  Take the growth rate $\mu(t)=e^{t}$, for all $t\in\mathbb{R}$. Define the nonuniform Lyapunov function  $S(t)=e^{2t}$, for all $t\in\mathbb{R}$. Clearly, we obtain the estimations $|S(t)|,|S(t)^{-1}|\leq \mu(t)^{\sgn(t)2}$, for all $t\in \mathbb{R}$. In addition, we get the identity  
			\[
			e^{-10(t-s)}e^{2s}=e^{2t}e^{-12(t-s)}, \qquad \text{ for all $t,s\in\mathbb{R}$}.
			\]
			
			Based on the preceding information, it follows that the earlier systems are nonuniformly $(\mu,2)$-kinematically similar. Nevertheless, note that
			\[
			\Sigma_\muD(-10)=\Sigma_\NmuD(-10)=\{-10\}\neq \{-12\}=\Sigma_\NmuD(-12)=\Sigma_\muD(-12).
			\] 
			
			Therefore, the $\mu$-dichotomy spectrum is not preserved.    
		}
	\end{example}
	
	In the next, inspired by \cite[Prop.~2.3]{BV}, we provide a second illustrative example of nonuniformly $(\mu,\varepsilon)$-kinematically similar systems which do not have the same nonuniform $\mu$-dichomoty spectra. Throughout the example, we will consider the growth rate as $\mu(t) = e^t$, for all $t \in \mathbb{R}$. 
	
	\begin{example}\label{634}
		{\rm
			Let $\varepsilon > 0$ and choose positive constants $\omega_1$, $\omega_2$ and $\kappa$ such that 
			\begin{equation}\label{636}
			\omega_1 > \omega_2 > 3\kappa \quad\text{ and }\quad \omega_1 - \omega_2 < \varepsilon.
			\end{equation}
			
			Consider the nonautonomous scalar systems
			\begin{equation}\label{626}
			\dot{u}=(-\omega_1-\kappa t\sin t)u,
			\end{equation}
			\begin{equation}\label{627}
			\dot{v}=(-\omega_2-\kappa t\sin t)v,
			\end{equation}
			with evolution operators given by
			\[
			\Phi(t,s)=e^{-\omega_1(t-s)+\kappa t \cos t-\kappa s \cos s-\kappa \sin t+\kappa \sin s},
			\]
			and
			\[
			\Psi(t,s)=e^{-\omega_2(t-s)+\kappa t \cos t-\kappa s \cos s-\kappa \sin t+\kappa \sin s},
			\]
			respectively. 
			
			Define the function $S(t)=e^{(\omega_2-\omega_1)t}$, for all $t\in\mathbb{R}$. Clearly, the function $S$ verifies:
			
			\begin{itemize}
				\item $|S(t)|\leq e^{(\omega_1-\omega_2)|t|}\leq e^{\varepsilon|t|}$, for all $t\in\mathbb{R}$.
				
				\item $|S(t)^{-1}|\leq e^{(\omega_1-\omega_2)|t|}\leq e^{\varepsilon|t|}$, for all $t\in\mathbb{R}$.
				
				\item $\Phi(t,s)S(s)=S(t)\Psi(t,s)$, for all $t,s\in\mathbb{R}$.
			\end{itemize}
			
			From Lemma~\ref{604}, we deduce that systems \eqref{626} and \eqref{627} are nonuniformly $(\mu,\varepsilon)$-kinematically similar.

			On the other hand, we have
			\begin{equation}\label{637}
			|\Phi(t,s)|\leq e^{2\kappa}e^{(-\omega_1+\kappa)(t-s)}e^{2\kappa|s|},\quad\,t\geq s. 
			\end{equation}
			
			From \cite[Prop.~2.3]{BV}, we infer that the parameters involved on the above estimation are optimal in the sense that $\st_{\Id}=-\omega_1+3\kappa$.  Similarly, we obtain 
			\begin{equation}\label{638}
			|\Phi(t,s)|\leq e^{2\kappa}e^{(-\omega_1-\kappa)(t-s)}e^{2\kappa|t|},\quad\,s\geq t,
			\end{equation}
			which in conjunction with \eqref{637} implies
			\[
			|\Phi(t,s)|\leq e^{2\kappa}e^{(\omega_1+3\kappa)|t-s|}e^{2\kappa|t|},\quad\forall\,s,t\in \mathbb{R}.
			\]
			
			In other words, the system \eqref{626} has nonuniform bounded growth, and thus its spectrum is compact and nonempty. Moreover, as \eqref{626} is one-dimensional, its spectrum is just one interval, hence it has the form $[a,b]$. Now we want to compute explicitly this spectrum.
			
			From Eq. \eqref{637}, the absolute value of the $\gamma$-shifted operator can be estimated as
			\[
			|\Phi_\gamma(t,s)|\leq e^{2\kappa}e^{(-\omega_1-\gamma+\kappa)(t-s)}e^{2\kappa|s|},\quad\,t\geq s.
			\]
			
			Therefore, if $\gamma>-\omega_1+3\kappa$, then $\gamma$ belongs to the nonuniform $\mu$-resolvent of the system \eqref{626} and $\st_\Id^\gamma=-\omega_1-\gamma+3\kappa$. Moreover, from Remark~\ref{639}, we deduce that $b=-\omega_1+3\kappa$.

			Similarly, from Eq. \eqref{638}, we derive the estimation 
			\[
			|\Phi_\gamma(t,s)|\leq e^{2\kappa}e^{(-\omega_1-\gamma-\kappa)(t-s)}e^{2\kappa|t|},\quad\,s\geq t.
			\]
			
			Therefore, if $\gamma<-\omega_1-3\kappa$, then $\gamma$ belongs to the nonuniform $\mu$-resolvent of the system \eqref{626} and $\un_0^\gamma=-\omega_1-\gamma+-3\kappa$. Now, from Remark \ref{639}, we infer that $a=-\omega_1-3\kappa$.

			Following this argument, it is deduced that the nonuniform $\mu$-dichotomy spectrum of system \eqref{626} is given by 
			\[
			\Sigma_\NmuD(u)=[-\omega_1-3\kappa,-\omega_1+3\kappa].
			\]
			
			Analogously to the previous analysis, we deduce that the nonuniform $\mu$-dichotomy spectrum of the system \eqref{627} is given by $\Sigma_\NmuD(v)=[-\omega_2-3\kappa,-\omega_2+3\kappa]$. In summary, $\Sigma_\NmuD(u) \neq \Sigma_\NmuD(v)$, notwithstanding the nonuniform $(\mu,\varepsilon)$-kinematic similarity between systems \eqref{626} and \eqref{627}.
			
			Finally, we emphasize that condition $\omega_1 > \omega_2 > 3\kappa$ is added just to ensure that the spectra are contained on the nonpositive semi axis, but the computations follow the same if we drop this assumption.
		}
	\end{example}
	
	\begin{remark}
		{\rm 
			A similar conclusion to \cite[Lemma 14]{Silva}, regarding the invariance of the nonuniform $\mu$-dichotomy spectrum, was presented in the nonuniform exponential case in \cite[Lemma 3.6]{Chu}, \cite[Corollary 3.7]{Chu}, and also in \cite[Lemma 3.2]{Xiang}. As observed in Example~\ref{ex1} and \ref{634}, there is an inconsistency in affirming the invariance of the nonuniform dichotomy spectrum through nonuniform kinematic similarity. This inconsistency directly challenges the assertions made in the aforementioned references, as well as any conclusions that relied on their assumed validity.  
		}
	\end{remark}
	
	\begin{remark}\label{641}
		{\rm In literature, see e.g. \cite{BV,Huerta}, some authors consider a slightly different definition of nonuniform dichotomy, where the estimations \eqref{600} are once again considered, but conditions $\alpha+\theta<0$ and $\beta-\nu>0$ are not required, while still taking $\alpha<0<\beta$ and $\theta,\nu\geq 0$. We call this a \textbf{slow nonuniform $\mu$-dichotomy} ($\sNmuD$). This dichotomy induces a different definition of the spectrum, but as every dichotomy is in particular a slow dichotomy, we obtain $\Sigma_\sNmuD\subset \Sigma_\NmuD\subset \Sigma_\muD$, and indeed the spectra can be different sets.
			
			For instance, consider the system \eqref{626} and the growth rate $\mu(t)=e^{t}$, for all $t\in\mathbb{R}$. Therefore, we have
			\[
			\Sigma_\sNmuD(u)=[-\omega_1-\kappa,-\omega_1+\kappa]\subset \Sigma_\NmuD(u)=[-\omega_1-3\kappa,-\omega_1+3\kappa]\subset \Sigma_\muD(u)=\mathbb{R}.
			\]
			
			Nevertheless, under a similar construction that we give here, it is proved that the slow spectrum is once again noninvariant under nonuniform kinematic similarities. Indeed, in our Example \ref{634}, we still have the nonuniform kinematic similarity between \eqref{626} and \eqref{627} while $\Sigma_\sNmuD(v)=[-\omega_2-\kappa,-\omega_2+\kappa]$, which shows that the problems with the spectra invariance do not come from the specific definition of dichotomy.}
	\end{remark}
	
	In the remainder of this section, we will attempt to propose a way to address the issue of noninvariance of the nonuniform $\mu$-dichotomy spectrum. For this purpose, we will employ the optimal stable and unstable ratio functions established in the previous section.

	\begin{definition}
		The function of \textbf{global optimal ratio} is $\gor:\rho_\NmuD\to \mathbb{R}$, defined by
		
		\begin{equation*} 
		\gor^\gamma=\left\{\begin{array}{lcc}
		\un_0^\gamma, &\text{ for }&\gamma\in (-\infty,a_1) ,\\
		\min\{-\st_\P^\gamma,\un_\P^\gamma\}, &\text{ for }&\gamma\in (b_i,a_{i+1}) ,\,i\in \{1,\dots,n-1\},\\
		-\st_\Id^\gamma, &\text{ for }&\gamma\in (b_n,+\infty).
		\end{array}
		\right.
		\end{equation*}
	\end{definition}
	
	The global optimal ratio function is defined on the $\mu$-resolvent set of a system having $(\Nmu,\epsilon)$-growth (which is open). This function is continuous (Proposition \ref{625}), strictly positive, tends to zero when $\gamma$ approaches to the spectrum (Theorem \ref{620}) and to infinity when $|\gamma|$ does (Lemma \ref{629}).
	
	\begin{proposition}\label{607}
		Assume the system \eqref{613} has  $(\Nmu,\epsilon)$-growth. Assume further that systems \eqref{613} and \eqref{601} are nonuniformly $(\mu,\varepsilon)$-kinematically similar. If $\gamma\in \rho_\NmuD(A)$ and
		\begin{equation*}
		3 \varepsilon< \gor^\gamma,
		\end{equation*}
		then $\gamma\in \rho_\NmuD(B)$.
	\end{proposition}
	
	\begin{proof}
		Assume without loss of generality that $\gamma$ lies on a bounded spectral gap, so the shifted system admits $\NmuD$ with some non-trivial projector $\P$. As the system  \eqref{613} has nonuniform $\mu$-bounded growth, the functions of optimal stable and unstable ratio $\st_\P$ and $\un_\P$ are well defined. Since $3\varepsilon<\gor^\gamma\leq-\st_\P^\gamma$, we can chose $\delta>0$ with $\delta< -\st_\P^\gamma-3\varepsilon$. By definition of $\st_\P^\gamma$, there is some $(\alpha_0,\theta_0)\in \St_\P^\gamma$ such that 
		\[\st_\P^\gamma<\alpha_0+\theta_0<\st_\P^\gamma+\delta<-3\varepsilon.\]
		Hence, we get $-(\alpha_0+\theta_0)>3\varepsilon$. Analogously we can find $(\beta_0,\nu_0)\in \Un_\P^\gamma$ such that $\beta_0-\nu_0>3\varepsilon$.

		From Corollary \ref{623} and Lemma \ref{605}, we infer that the $\gamma$-shifted system associated to \eqref{601} also admits $\NmuD$, implying $\gamma\in \rho_\NmuD(B)$.
	\end{proof}
	
	In the next, we shall proceed to define a novel notion of nonuniform $\mu$-dichotomy spectrum. This concept aims to offer a more effective approach in addressing the specific challenge posed by the noninvariance through nonuniform  $(\mu,\varepsilon)$-kinematic similarity illustrated on Examples \ref{ex1} and \ref{634}.
	
	\begin{definition}\label{epsilonspectra}
		Assume the system \eqref{613} has $(\Nmu,\epsilon)$-growth.  For every $\varepsilon\geq 0$, the \textbf{$\varepsilon$-interior of the nonuniform $\mu$-resolvent} for system \eqref{613} is the set defined by
		$$\rho_\NmuD^\varepsilon(A):=\left\{\gamma\in \rho_\NmuD: \gor^\gamma>\varepsilon\right\},$$
		and the \textbf{$\varepsilon$-neighborhood of the nonuniform $\mu$-dichotomy spectrum} for system \eqref{613} is the set defined by 
		\[
		\Sigma_\NmuD^\varepsilon(A):=\mathbb{R}\setminus\rho_\NmuD^\varepsilon(A).
		\]
	\end{definition}
	
	Note that for $\varepsilon=0$, we have $\rho_\NmuD^\varepsilon(A)=\rho_\NmuD(A)$ and $\Si_\NmuD^\varepsilon(A)=\Si_\NmuD(A)$.
	
	\begin{proposition}
		Assume the system \eqref{613} has $(\Nmu,\epsilon)$-growth. For every $\varepsilon>0$, the $\varepsilon$-neighborhood of the nonuniform $\mu$-dichotomy spectrum is a finite union of compact intervals. Moreover, it is indeed a neighborhood of the closed set $\Sigma_\NmuD(A)$,  i.e. there is an open set $U$ such that $\Sigma_\NmuD(A)\subset U\subset \Sigma^\varepsilon_\NmuD(A)$.  
	\end{proposition}
	
	\begin{proof}
		As mentioned earlier, the function $\gor$ is continuous and defined on the open set $\rho_\NmuD(A)$, thus the $\varepsilon$-interior of the nonuniform $\mu$-resolvent is by definition an open set contained on $\rho_\NmuD(A)$, which implies that $\Sigma_\NmuD^\varepsilon(A)$ is closed. Moreover, by Corollary \ref{621}, on every bounded spectral gap $(b_i,a_{i+1})$ there are two positive numbers $\delta_1,\delta_2>0$ such that 
		\begin{equation}\label{628}
		\rho_\NmuD^\varepsilon(A)\cap (b_i,a_{i+1})=(b_i+\delta_1,a_{i+1}-\delta_2),
		\end{equation}
		where in the case that $a_{i+1}-\delta_2<b_i+\delta_1$, we consider the right-hand  side of \eqref{628} as the empty set.

		Now, for the spectral gaps $(-\infty,a_1)$ and $(b_n,+\infty)$, from Lemma \ref{629}, we know both $-\st_\P$ and $\un_\P$ tend to $+\infty$ when $|\gamma|\to +\infty$, thus there are two positive numbers $\delta_3,\delta_4>0$ such that
		\[\rho_\NmuD^\varepsilon(A)\cap (-\infty,a_{1})=(-\infty,a_{i+1}-\delta_3),\] 
		and 
		\[\rho_\NmuD^\varepsilon(A)\cap (b_n,+\infty)=(b_n+\delta_4,+\infty),\]
		which implies that $\Sigma_\NmuD(A)^\varepsilon$ is bounded, and thus compact. Moreover, the previous intersections, in conjunction to \eqref{628}, imply that $\Sigma_\NmuD^\varepsilon(A)$ is indeed a neighborhood of $\Sigma_\NmuD(A)$.
	\end{proof}
	
	We finalize this section by presenting an important consequence from Proposition \ref{607}. In this statement, we just assume $(\Nmu,\epsilon)$-growth for system \eqref{613}, since this property is preserved through nonuniform $(\mu,\varepsilon)$-kinematic similarity.
	
	\begin{corollary}\label{608}
		Assume the system \eqref{613} has $(\Nmu,\epsilon)$-growth. If the systems \eqref{613} and \eqref{601} are nonuniformly $(\mu,\varepsilon)$-kinematically similar, then 
		
		\[
		\rho_\NmuD^{3\varepsilon}(A)\subset \rho_\NmuD(B) \quad\text{ and }\quad\rho_\NmuD^{3\varepsilon}(B)\subset \rho_\NmuD(A),
		\]
		or equivalently,
		\[
		\Sigma_\NmuD(A)\subset \Sigma_\NmuD^{3\varepsilon}(B) \quad\text{ and }\quad \Sigma_\NmuD(B)\subset \Sigma_\NmuD^{3\varepsilon}(A).
		\] 
	\end{corollary}

	\begin{remark}
		{\rm
			Note that the previous corollary is consistent with Examples \ref{ex1} and \ref{634}. For instance, let us consider the systems \eqref{626} and \eqref{627}. Since $\omega_1-\omega_2<\varepsilon$, we infer the following inclusion
			\[\Sigma_\NmuD(u)=[-\omega_1-3\kappa,-\omega_1+3\kappa]\subset [-\omega_2-3\varepsilon-3\kappa,-\omega_2+3\varepsilon+3\kappa]=\Sigma_\NmuD^{3\varepsilon}(v)\,. \]
		}
	\end{remark}

	\section{Endgame remarks}
	
	In this section, we present a synthesis of final remarks derived from our main results: 
	
	\begin{itemize}
		\item In light of the preceding Corollary~\ref{608}, we are unable to establish, using the available tools, that if systems \eqref{613} and \eqref{601} are nonuniformly $(\mu,\varepsilon)$-kinematically similar, then $\rho_\NmuD(A) = \rho_\NmuD(B)$, nor even $\rho_\NmuD^{3\varepsilon}(B) = \rho_\NmuD^{3\varepsilon}(A)$.
		
		\smallskip
		
		\item In \cite[Lemma 14]{Silva}, it was stated that the demonstration we give here for Lemma~\ref{605} is enough to conclude equality of the $\mu$-resolvent sets, or equivalently, the equality of the nonuniform $\mu$-dichotomy spectra. However, our Theorem~\ref{620} reveals that for every $\varepsilon>0$, the set $\rho_\NmuD^\varepsilon(A)$ is strictly contained in $\rho_\NmuD(A)$, hence the best possible conclusion from Lemma \ref{605} is our Proposition \ref{607} or our Corollary \ref{608}, {\it i.e.} contention of the $\varepsilon$-interior, but not equality of the resolvent sets or spectra. Consequently, our results highlight certain inconsistencies in the proofs of \cite[Lemma 14]{Silva}, \cite[Corollary 3.7]{Chu} and \cite[Lemma 3.2]{Xiang}, where the assertion of the invariance of the nonuniform dichotomy spectrum through nonuniform kinematic similarity is affirmed. Thereupon, our findings call into question as well the results that depend on the assumption of the spectrum invariance. It could be interesting to revisit them.
		
		\smallskip
		
		\item The result about the invariance of uniform exponential dichotomy espectra under classic uniform kinematic similarity \cite[Corollary~2.1]{Siegmund2} still holds, since in the case of a classic kinematic similarity, we have a Lyapunov matrix function verifying (i) and (ii) from Definition~\ref{Kinsim} with $\varepsilon=0$. Thus, the same proof we give here for Lemma~\ref{605} shows the spectra invariance. In addition, we can provide an invariance result of the nonuniform $\mu$-dichotomy spectrum considering uniform kinematically similar systems:
		\begin{corollary}\label{640}
			Assume the systems \eqref{613} and \eqref{601} are uniformly kinematically similar. Then $\rho_\NmuD(A)= \rho_\NmuD(B)$, or equivalently, $\Sigma_\NmuD(B)=\Sigma_\NmuD(A)$.
		\end{corollary}
		\begin{proof}
			It follows from Lemma \ref{605}.
		\end{proof}

		\begin{example}
			{\rm  Define the functions $\mathfrak{a}(t)=-\omega-\kappa t \sin t$, and $\mathfrak{b}(t)=-\kappa t\sin t$, for all $t\in\mathbb{R}$, where $\omega>6\kappa>0$. Consider the planar equations $\dot{x}=A(t)x$ and $\dot{y}=B(t)y$, where
				\[
				A(t)=\begin{pmatrix}
				\mathfrak{a}(t) & 0\\
				0& \mathfrak{b}(t)
				\end{pmatrix}, \qquad t\in\mathbb{R},
				\]
				and 
				\[
				B(t)=\begin{pmatrix}
				\mathfrak{a}(t)\cos^2 t+\mathfrak{b}(t)\sin^2 t& \omega\cos t \sin t+1\\
				\omega\cos t\sin t-1& \mathfrak{a}(t)\sin^2t+\mathfrak{b}(t)\cos^2t
				\end{pmatrix}, \qquad t\in\mathbb{R}.
				\]
				
				The evolution operators $\Phi_A$ and $\Psi_B$ are given by
				\[\Phi_A(t,s)=\begin{pmatrix}
				e^{-\omega(t-s)+k(t,s)} & 0\\
				0& e^{k(t,s)}
				\end{pmatrix},
				\]
				where $k(t,s)=\kappa t \cos t-\kappa s \cos s-\kappa \sin t+\kappa \sin s$, for all $t,s\in\mathbb{R}$,   and
				\[
				\Psi_B(t,s)=\begin{pmatrix}
				a_{11}(t,s)&a_{12}(t,s)\\
				a_{21}(t,s)& a_{22}(t,s)
				\end{pmatrix},  \qquad t,s\in\mathbb{R},
				\]
				where
				\begin{itemize}
					\item[] $a_{11}(t,s)=\cos t  \cos s\,e^{-\omega(t-s)+k(t,s)} +\sin t\sin s\,e^{k(t,s)}$,
					\item[] $a_{12}(t,s)=-  \cos t\sin s\,e^{-\omega(t-s)+k(t,s)} +\sin t\cos s\,e^{k(t,s)}$,
					\item[] $a_{21}(t,s)=- \sin t\cos s\,e^{-\omega(t-s)+k(t,s)} +\cos t\sin s\,e^{k(t,s)}$,
					\item[] $a_{22}(t,s)=\sin t  \sin s\, e^{-\omega(t-s)+k(t,s)} +\cos t\cos s\,e^{k(t,s)}$.
				\end{itemize}
				
				Let us define
				\[
				S(t)=\begin{pmatrix}
				\cos t & -\sin t\\
				\sin t& \cos t
				\end{pmatrix}, \qquad t\in\mathbb{R}.
				\]
				Clearly, $\|S(t)\|=\|S(t)^{-1}\|=1$, for all $t\in \mathbb{R}$. In addition, the identity
				\[
				S(t)\Psi_B(t,s)=\Phi_A(t,s)S(s)
				\]
				holds, for all $t,s\in\mathbb{R}$. In consequence, the planar systems are uniformly kinematically similar.

				Proceeding analogously as we computed the spectrum on Example \ref{634}, it is deduced that $\Sigma_\NmuD(A)=[-\omega-3\kappa,-\omega+3\kappa]\cup[-3\kappa,3\kappa]$. Thus, by virtue of Corollary~\ref{640}, we obtain $\Sigma_\NmuD(B)=[-\omega-3\kappa,-\omega+3\kappa]\cup[-3\kappa,3\kappa]$.
			}
		\end{example}
		
		\item For a similar reason, Silva's reducibility result \cite[Theorem 12]{Silva} still holds, but only for the uniform $\mu$-dichotomy case, that is, when the system verifies uniform $\mu$-bounded growth.
		
		\smallskip
		
		\item Even if the spectrum can change, by the continuity of the global optimal ratio map $\gor$, some characteristics can still be held for nonuniformly kinematically similar systems. Consider for instance the following conclusions immediately obtained from Corollary \ref{608}.
		
		\begin{corollary}
			Assume the system \eqref{613} has $(\Nmu,\epsilon)$-growth and verifies that $\Sigma_\NmuD(A)\subset (-\infty,0)$. Then, there is some $\varepsilon>0$ such that every system which is nonuniformly $(\mu,\varepsilon)$-kinematically similar to \eqref{613} also has its spectrum contained on $(-\infty,0)$.
		\end{corollary}
		
		The previous conclusion is relevant since the condition of having the spectrum contained in one of the semiaxis is a usual characterization of stable or unstable systems. A similar strategy can be employed to prove the following result regarding the notion of nonresonant spectra \cite{Jara5,Cuong,Siegmund3}, which, as mentioned on the introduction, has important applications in the realm of linearization of nonautonomous systems.
		
		\begin{corollary}
			Assume the system \eqref{613} has $(\Nmu,\epsilon)$-growth and its spectrum $\Si_\NmuD(A)=[a_1,b_1]\cup\cdots\cup [a_n,b_n]$ is nonresonant up to a certain order $p\in \mathbb{N}$, i.e. for every
			$j=1,\dots,n$ and every $(k_1,\dots,k_n)\in (\mathbb{N}\cup\{0\})^n$ with $2\leq \sum_{i=1}^nk_i\leq p$, we have
			\begin{equation*}
			[a_j,b_j]\cap \left[\sum_{i=1}^nk_ia_i,\sum_{i=1}^nk_ib_i\right]=\emptyset,
			\end{equation*}
			then, there is some $\varepsilon>0$ such that every system which is nonuniformly $(\mu,\varepsilon)$-kinematically similar to \eqref{613} also has a nonresonant spectrum up to the order $p\in \mathbb{N}$.
		\end{corollary}
		
		\item A question we leave open is the rate at which the convergence in Theorem \ref{620} is attained. Note that in all our examples we have the estimations $|\st_\P^\gamma|\leq |\gamma-b_i|$ and $\un_\P^\gamma\leq |a_{i+1}-\gamma|$, for $\gamma$ in the correspondent spectral gap. We conjecture that this estimation is always valid,  but we have not been able to prove it or find a counterexample.
	\end{itemize}


\begin{thebibliography}{99}
		
		
		
		\bibitem{Dragicevic3}
		Backes, L.; Dragi\v{c}evi\'c, D.: {\it Smooth linearization of nonautonomous dynamics under polynomial behaviour}, preprint available at 
		https://doi.org/10.48550/arXiv.2210.04804.
		
		\bibitem{BV} Barreira L.; Valls, C.: Stability of Nonautomous Differential Equations, Lecture Notes in Mathematics, Springer, 2008.
		
		\bibitem{Barreira2}
		Barreira, L.; Valls, C.: {\it A Grobman-Hartman theorem for nonuniformly hyperbolic dynamics}, J. Differential Equations 228 (2006), no. 1, 285–310.
		
		\bibitem{Barreira7}
		Barreira, L; Valls, C.: {\it Smoothness of invariant manifolds for nonautonomous equations}, Comm. Math. Phys. 259 (2005), no.3, 639–677.
		
		\bibitem{Castaneda6}
		Casta\~{n}eda, \'A.; Huerta, I.; Robledo, G.: {\it The dichotomy spectrum approach for a global nonuniform asymptotic stability problem: triangular case via uniformization}, preprint available at https://doi.org/10.48550/arXiv.2210.01943.
		
		\bibitem{Jara5}
		Castañeda, Á.; Jara, N.: {\it A generalization of Siegmund's normal forms theorem to systems with $\mu$-dichotomies}, preprint available at 
		https://doi.org/10.48550/arXiv.2312.04544.
		
		
		\bibitem{Coppel} Coppel, W. A.: Dichotomies in Stability Theory, Lecture Notes in Mathematics, vol. 629, Springer, 1978.
		
		\bibitem{Chu}
		Chu, J.; Liao, F-F.; Siegmund, S.; Xia, Y.; Zhang, W.: {\it Nonuniform dichotomy spectrum and reducibility for nonautonomous equations}, Bull. Sci. Math. 139 (2015), no.5, 538–557.
		
		
		\bibitem{Cuong}
		Cuong, L. V.; Doan, T. S.; Siegmund, S.: {\it A Sternberg theorem for nonautonomous differential equations}, J. Dynam. Differential Equations 31 (2019), no. 3, 1279–1299.
		
		
		\bibitem{Dragicevic2}
		Dragi\v{c}evi\'c, D.; Zhang, W., Zhang, W.:
		{\it Smooth linearization of nonautonomous differential equations with a nonuniform dichotomy}, Proc. Lond. Math. Soc. 121 (2020), 32--50.
		
		
		\bibitem{Dragicevic5}
		Dragi\v{c}evi\'c, D.: {\it Admissibility and nonuniform polynomial dichotomies}, Math. Nachr. 293 (2020), 226–243.
		
		\bibitem{Dragicevic6}
		Dragi\v{c}evi\'c, D.;  Sasu, A. L.; Sasu, B.: {\it On polynomial dichotomies of discrete nonautonomous
			systems on the half-line}, Carpathian J. Math. 38 (2022), 663–680.
		
		
		\bibitem{Huerta}
		Huerta, I.: {\it Linearization of a nonautonomous unbounded system with nonuniform contraction: a spectral approach}, Discrete Contin. Dyn. Syst. 40(2020), no.9, 5571–5590.
		
		\bibitem{Perron}
		Perron, O.: {\it Die Stabilitätsfrage bei Differentialgleichungen}, Math. Z. 32 (1930), no.1, 703–728.
		
		\bibitem{Poincare}
		Poincar\'e, H.: Mémoire sur les courbes définies par une équation différentielle, Thesis 5 (1879); Also Oeuvres I, Gauthier---Villars, Paris, 1928, pp. 59--129.
		
		
		\bibitem{Sacker}
		Sacker, R. J.; Sell, G. R.: {\it A spectral theory for linear differential systems}, J. Differential Equations 27 (1978), no. 3, 320–358.
		
		\bibitem{Silva}
		Silva, C. M.: {\it Nonuniform $\mu$-dichotomy spectrum and kinematic similarity}, J. Differential Equations 375 (2023), 618-652.
		
		
		\bibitem{Siegmund}
		Siegmund, S.: {\it Dichotomy spectrum for nonautonomous differential equations}, J. Dynam. Differential Equations 14 (2002), no.1, 243–258
		
		\bibitem{Siegmund2}
		Siegmund, S.: {\it Reducibility of nonautonomous linear differential equations}, J. London Math. Soc. (2) 65 (2002), no. 2, 397–410.
		
		\bibitem{Siegmund3}
		Siegmund, S.: {\it Normal forms for nonautonomous differential equations}, J. Differential Equations 178 (2002), no.2, 541–573.
		
		\bibitem{Sternberg1}
		Sternberg, S.: {\it Local contractions and a theorem of Poincar\'e}, Am. J. Math. 79 (1957), 809-824.
		
		\bibitem{Sternberg2}
		Sternberg, S.: {\it On the structure of local homeomorphisms of Euclidian $n-$space}, I. Am. J. Math. 80 (1958), 623-631.
		
		
		\bibitem{Xiang}
		Zhang, X.: {\it Nonuniform dichotomy spectrum and normal forms for nonautonomous differential systems}, J. Funct. Anal. 267 (2014), no.7, 1889–1916.
		
		
		
	\end{thebibliography}
\end{document}